\documentclass{article}
\usepackage{fullpage}
\usepackage{graphicx}
\usepackage{amsmath, amsthm, amssymb, cite, color, comment, mathtools, url}
\usepackage[unicode]{hyperref}

\newtheorem{theorem}{Theorem}[section]
\newtheorem{lemma}[theorem]{Lemma}

\newtheorem{observation}[theorem]{Observation}

\newtheorem{claim}[theorem]{Claim}

\theoremstyle{definition}

\newtheorem*{problem}{Problem}

\theoremstyle{remark}
\newtheorem{remark}[theorem]{Remark}

\newcommand{\In}{\mathrm{in}}
\newcommand{\Out}{\mathrm{out}}
\newcommand{\End}{\mathrm{end}}

\title{Finding Spanning Trees with Perfect Matchings}
\author{
  Krist\'{o}f B\'{e}rczi\thanks{E\"{o}tv\"{o}s Lor\'{a}nd University, Hungary. \texttt{kristof.berczi@ttk.elte.hu}} \and \!\!\!\!
  Tam\'{a}s Kir\'{a}ly\thanks{E\"{o}tv\"{o}s Lor\'{a}nd University, Hungary. \texttt{tamas.kiraly@ttk.elte.hu}} \and \!\!\!\!
  Yusuke Kobayashi\thanks{Kyoto University, Japan. Email: \texttt{yusuke@kurims.kyoto-u.ac.jp}} \and \!\!\!\!
  Yutaro Yamaguchi\thanks{Osaka University, Japan. Email: \texttt{yutaro.yamaguchi@ist.osaka-u.ac.jp}} \and \!\!\!\!
  Yu Yokoi\thanks{Tokyo Institute of Technology, Japan. Email: \texttt{yokoi@c.titech.ac.jp}}}
\date{\empty}

\begin{document}
\maketitle
\thispagestyle{empty}

\begin{abstract}
We investigate the tractability of a simple fusion of two fundamental structures on graphs, a spanning tree and a perfect matching.
Specifically, we consider the following problem: given an edge-weighted graph, find a minimum-weight spanning tree among those containing a perfect matching.
On the positive side, we design a simple greedy algorithm for the case when the graph is complete (or complete bipartite) and the edge weights take at most two values.
On the negative side, the problem is NP-hard even when the graph is complete (or complete bipartite) and the edge weights take at most three values, or when the graph is cubic, planar, and bipartite and the edge weights take at most two values.

We also consider an interesting variant.
We call a tree strongly balanced if on one side of the bipartition of the vertex set with respect to the tree, all but one of the vertices have degree $2$ and the remaining one is a leaf.
This property is a sufficient condition for a tree to have a perfect matching, which enjoys an additional property. 
When the underlying graph is bipartite, strongly balanced spanning trees can be written as matroid intersection, and this fact was recently utilized to design an approximation algorithm for some kind of connectivity augmentation problem.
The natural question is its tractability in nonbipartite graphs.
As a negative answer, it turns out NP-hard to test whether a given graph has a strongly balanced spanning tree or not even when the graph is subcubic and planar.
\end{abstract}


\clearpage
\setcounter{page}{1}

\section{Introduction}
Spanning trees and perfect matchings are two of the most fundamental tractable structures in combinatorial optimization on graphs.
It is well-known that a minimum-weight spanning tree or perfect matching in an edge-weighted graph can be found in polynomial time~\cite{kruskal1956shortest,prim1957shortest,dijkstra1959note,edmonds1965maximum}.
These problems are foundational for the development of not only graph algorithms but also the theory of combinatorial optimization on more abstract targets such as matroids and discrete convex functions~(cf.~\cite{frank2011connections,murota2003discrete,fujishige2005submodular,schrijver2003combinatorial}).

In this paper, we investigate a kind of fusion of these two structures.
The simplest setting is to determine whether a given graph has a spanning tree containing a perfect matching or not.
This problem is easy, because a trivial necessary condition that the graph is connected and has a perfect matching is also sufficient (Observation~\ref{obs:PMST}).

The natural question is as follows: what about finding a minimum-weight such spanning tree in an edge-weighted graph?
We give a solid answer to this question (Theorem~\ref{thm:PMST}).
On the positive side, this optimization problem is tractable in the very restricted situation when the graph is a complete or complete bipartite graph and the edge weights are restricted to at most two values.
On the negative side, the problem is NP-hard even if the graph is a complete or complete bipartite graph and the edge weights are restricted to at most three values (say, $0$, $1$, or $2$), or if the graph is cubic, planar, and bipartite and the edge weights are restricted to at most two values (say, $0$ or $1$).

We also consider a variant of a spanning tree with a perfect matching.
A trivial necessary condition for a tree to have a perfect matching is that the bipartition of the vertex set is balanced;
that is, if we color the vertices with two colors so that any two adjacent vertices are colored differently, then the number of vertices colored with each of the two colors is the same.
This is clearly not sufficient, but it can be strengthened to become sufficient as follows.

We say that a tree is \emph{strongly balanced} if on one side of the bipartition, exactly one vertex is a leaf (that is of degree $1$) and all the other vertices are of degree $2$.
Norose and Yamaguchi~\cite{norose2024approximation} showed that a tree is strongly balanced if and only if it has a perfect matching and, for some leaf, a path from the leaf to any vertex in the tree is an alternating path with respect to the perfect matching.
This property was utilized to design a nontrivial approximation algorithm for a kind of connectivity augmentation problem~\cite{norose2024approximation} and travelling salesman problem~\cite{frank1998bipartite}.
An important fact is that, for a bipartite graph, strongly balanced spanning trees can be represented as the common bases of two matroids (Observation~\ref{obs:SBST}), which enables us to find a minimum-weight strongly balanced spanning tree in polynomial time with the aid of weighted matroid intersection algorihtms \cite{lawler1970optimal,lawler1975matroid,edmonds1979matroid,iri1976algorithm,brezovec1986two}.
Also, as pointed out in \cite{norose2024approximation}, it is interesting that this problem in fact commonly generalizes two fundamental special cases of weighted matroid intersection: finding a minimum-weight perfect matching in bipartite graphs and finding a minimum-weight arborescence in directed graphs.

The natural question, again, is as follows: what about the tractability of strongly balanced spanning trees in nonbipartite graphs?
We give a negative answer to this question (Theorem~\ref{thm:SBST}):
it is NP-hard to test whether a given graph has a strongly balanced spanning tree or not, even if the graph is subcubic and planar.

Problems of finding a spanning tree with additional constraints have been studied extensively, some of which were motivated by applications to communication networks.
For example, there are several studies on spanning trees with degree bounds \cite{Czumaj:Strothmann:ESA:1997,FURER1994409,DBLP:conf/focs/Goemans06,10.1145/2629366}
and spanning trees with small diameter~\cite{10.1016/0020-0190(94)00183-Y,HASSIN2004343,DBLP:journals/algorithmica/SpriggsKBSS04}.  
Our problems also align with this context. 
The additional condition of containing a perfect matching implies that each node can be paired with another to have mutual backup, expressing network robustness in a sense.
It is worth noting that the constraints of having partners and backups were studied 
also in the dominating set problem~\cite{DBLP:journals/networks/HaynesS98,DBLP:journals/jco/ChenLZ10,DBLP:journals/tcs/TripathiKPPW22}. 
Spanning trees with perfect matchings also appear in characterization of chemical structures \cite{wu2021graphs,vukiveevic2007anti}, and counting them in special graphs has recently been paid attention~\cite{ma2024enumeration}.

The rest of the paper is organized as follows.
In Section~\ref{sec:preliminaries}, we describe necessary definitions, and formally state the problems and our results.
In Section~\ref{sec:PMST}, we give a simple, polynomial-time algorithm for finding a minimum-weight spanning tree containing a perfect matching for the restricted inputs, and show the NP-hardness of the problem by a reduction from the Hamiltonian cycle problem.
In Section~\ref{sec:SBST}, we prove the NP-hardness of finding a strongly balanced spanning tree by a reduction from the planar 3-SAT problem.
In Section~\ref{sec:conclusion}, we conclude the paper with possible future work.

\section{Preliminaries}\label{sec:preliminaries}
\subsection{Definitions}
Let $G = (V, E)$ be a graph, which we assume to be simple and undirected unless otherwise specified.
We refer the readers to \cite{schrijver2003combinatorial} for basic concepts and notation on graphs.

An edge set $M \subseteq E$ is a \emph{matching} in $G$ if the edges in $M$ do not share their end vertices.
For a fixed matching, a vertex is said to be \emph{covered} (or \emph{matched}) if it is an end vertex of an edge in the matching, and \emph{exposed} otherwise.
A matching is \emph{perfect} if all the vertices are covered.
Let $\mathrm{def}(G)$ denote the \emph{deficiency} of $G$, which is defined as the minimum number of vertices exposed by a matching in $G$; in other words, it is the number of vertices minus twice the maximum cardinality of a matching.

A graph is said to be \emph{bipartite} if its vertex set admits a \emph{bipartition} $(V^+, V^-)$ such that every edge has one of its end vertices in $V^+$ and the other in $V^-$; possibly $V^+ = \emptyset$ or $V^- = \emptyset$.
We say that a bipartite graph (with a fixed bipartition) is \emph{balanced} if $|V^+| = |V^-|$.
Note that if a bipartite graph has a perfect matching, then it is balanced (regardless of the bipartition).

Let $T$ be a spanning tree of $G$.
A tree is bipartite, and let $(V_T^+, V_T^-)$ denote a bipartition of $V$ with respect to $T$, i.e., $V_T^+ \cup V_T^- = V$, $V_T^+ \cap V_T^- = \emptyset$, and every edge in $T$ connects a vertex in $V_T^+$ and one in $V_T^-$.
Such a bipartition is unique up to the symmetry of $V_T^+$ and $V_T^-$.
Note that if $T$ has a perfect matching, then it is unique and $T$ is balanced.
We say that $T$ is \emph{strongly balanced} if on one side of the bipartition, say $V_T^+$ by symmetry, exactly one vertex is a leaf (that is of degree $1$) and all the other vertices are of degree $2$.
Observe that if $T$ is strongly balanced, then $T$ is balanced.
As a relation with a perfect matching, the following characterization is known.

\begin{lemma}[Norose and Yamaguchi~\cite{norose2024approximation}]\label{lem:NY2024}
For a tree $T$, the following two statements are equivalent.
\begin{itemize}
\setlength{\itemsep}{0mm}
	\item $T$ is strongly balanced.
	\item $T$ has a perfect matching $M$, and there exists a leaf $r$ such that a path in $T$ from $r$ to any vertex alternates between the edges in $M$ and in $E \setminus M$.
\end{itemize}
\end{lemma}

Suppose that $G$ is associated with edge weight $w \colon E \to \mathbb{R}$.
We define the \emph{weight} of each edge set $F \subseteq E$ as $w(F) \coloneq \sum_{e \in F} w(e)$.

\subsection{Problems and Our Results}
We are now ready to state the problems.
In the following sections, those problems are referred to as their abbreviated names.

\smallskip

\begin{problem}[\textsc{Perfectly Matchable Spanning Tree (PMST)}]
\mbox{ }
\begin{description}
\setlength{\itemsep}{0mm}
\item[Input:] A graph $G = (V, E)$.
\item[Goal:] Decide whether $G$ has a spanning tree containing a perfect matching or not.
\end{description}
\end{problem}

\smallskip

\begin{problem}[\textsc{Minimum Perfectly Matchable Spanning Tree (MinPMST)}]
\mbox{ }
\begin{description}
\setlength{\itemsep}{0mm}
\item[Input:] A graph $G = (V, E)$ with edge weight $w \colon E \to \mathbb{R}$.
\item[Goal:] Minimize $w(T)$ subject to $T$ is a spanning tree of $G$ containing a perfect matching.
\end{description}
\end{problem}

\smallskip

\begin{problem}[\textsc{Strongly Balanced Spanning Tree (SBST)}]
\mbox{ }
\begin{description}
\setlength{\itemsep}{0mm}
\item[Input:] A graph $G = (V, E)$.
\item[Goal:] Decide whether $G$ has a strongly balanced spanning tree $T$ or not.
\end{description}
\end{problem}

\smallskip

\begin{problem}[\textsc{Minimum Strongly Balanced Spanning Tree (MinSBST)}]
\mbox{ }
\begin{description}
\setlength{\itemsep}{0mm}
\item[Input:] A graph $G = (V, E)$ with edge weight $w \colon E \to \mathbb{R}$.
\item[Goal:] Minimize $w(T)$ subject to $T$ is a strongly balanced spanning tree of $G$.
\end{description}
\end{problem}

\smallskip
The main results shown in this paper are summarized as follows.

\begin{theorem}\label{thm:PMST}
\mbox{ }
\begin{enumerate}
\setlength{\itemsep}{0mm}
	\item \textsc{MinPMST} can be solved in polynomial time if $G$ is a complete or complete bipartite graph and $|\mathrm{cod}(w)| \le 2$, where $\mathrm{cod}(w) \coloneq \{\, w(e) \mid e \in E \,\}$ denotes the codomain of $w$.
	\item \textsc{MinPMST} is NP-hard even if the input is restricted as follows:
	\begin{enumerate}
	\setlength{\itemsep}{1mm}
		\item $G$ is a cubic planar bipartite graph and $|\mathrm{cod}(w)| \le 2$.
		\item $G$ is a complete or complete bipartite graph and $|\mathrm{cod}(w)| \le 3$.
	\end{enumerate}
\end{enumerate}
\end{theorem}

\begin{theorem}\label{thm:SBST}
\textsc{SBST} is NP-hard even if $G$ is restricted to a subcubic planar graph.
\end{theorem}

\section{On Perfectly Matchable Spanning Trees (Proof of Theorem~\ref{thm:PMST})}\label{sec:PMST}
In this section, we prove Theorem~\ref{thm:PMST}.
We start with an easy observation, which immediately leads to the tractability of \textsc{PMST} with the aid of polynomial-time algorithms for finding a maximum matching in graphs~\cite{edmonds1965paths}; just find a perfect matching, and make it connected by adding edges between different connected components, one-by-one.

\begin{observation}\label{obs:PMST}
A graph $G$ has a spanning tree containing a perfect matching if and only if $G$ is connected and has a perfect matching.
\end{observation}

In Section~\ref{sec:PMST_tractable}, we prove Statement~1 by giving a solution to \textsc{MinPMST} when the graph is a complete or complete bipartite graph and there are at most two weight values.
In Section~\ref{sec:PMST_hard}, we prove Statement~2 by giving a reduction from the Hamiltonian cycle problem with an input restriction to \textsc{MinPMST} with the stated input restrictions.

\subsection{Tractable Case (Statement 1)}\label{sec:PMST_tractable}
Since all the spanning trees of a fixed graph have the same number of edges, the case when $|\mathrm{cod}(w)| = 1$ reduces to \textsc{PMST}.
Suppose that $|\mathrm{cod}(w)| = 2$.
Then, the objective is rephrased as minimizing the number of heavier edges.
Thus, without loss of generality, we assume that $\mathrm{cod}(w) = \{0, 1\}$.

Let us focus on the complete graph case (the complete bipartite graph case is almost the same; see Remark~\ref{rem:bipartite} at the end of this section).
Let $G = (V, E)$ be a complete graph with edge weight $w \colon E \to \{0, 1\}$, where we assume $|V|$ is even (otherwise, $G$ has no perfect matching).
Let $G_0 = (V, E_0)$ be the subgraph consisting of all the edges of weight $0$, i.e., $E_0 = \{\, e \in E \mid w(e) = 0 \,\}$.
Then, by Observation~\ref{obs:PMST}, the following augmentation problem with input $H = G_0$ is equivalent to \textsc{MinPMST} with the current restriction.

\begin{problem}[\textsc{Augmentation on PMST (AugPMST)}]
\mbox{ }
\begin{description}
\setlength{\itemsep}{0mm}
\item[Input:] A graph $H$.
\item[Goal:] Minimize the number of additional edges to make $H$ connected and perfectly matchable.
\end{description}
\end{problem}

For \textsc{AugPMST}, we show a complete characterization as follows, which leads to a solution to the original problem \textsc{MinPMST} when $G$ is a complete graph and $|\mathrm{cod}(w)| = 2$.
For a graph $H$, we denote by $\mathrm{opt}(H)$ the optimal value for the input $H$ of \textsc{AugPMST}.
Also, let $c(H)$ denote the number of connected components of $H$, and $c_0(H)$ and $c_+(H)$ denote the numbers of connected components $K$ with $\mathrm{def}(K) = 0$ and with $\mathrm{def}(K) > 0$, respectively.

\begin{lemma}\label{lem:PMST}
\mbox{ }
\begin{enumerate}
\setlength{\itemsep}{0mm}
	\item If $\mathrm{def}(H) = 0$, then $\mathrm{opt}(H) = c(H) - 1$.
	\item Suppose that $\mathrm{def}(H) > 0$.
	\begin{enumerate}
	\setlength{\itemsep}{1mm}
		\item If $\frac{1}{2}\mathrm{def}(H) < c_+(H)$, then $\mathrm{opt}(H) = c(H) - 1$.
		\item If $\frac{1}{2}\mathrm{def}(H) \ge c_+(H)$, then $\mathrm{opt}(H) = \frac{1}{2}\mathrm{def}(H) + c_0(H)$.
	\end{enumerate}
\end{enumerate}
\end{lemma}

\begin{proof}
We show this lemma by induction on the number of edges in $H$ (in descending order).
The proof is based on case analysis with respect to $\varphi(H) \coloneq \left(\frac{1}{2}\mathrm{def}(H), c_+(H), c_0(H)\right)$.

We first observe all the possible changes of $\varphi(H)$ by adding an edge to $H$.
Suppose that $\varphi(H) = (i, j, k)$.
Let $H''$ be any graph obtained from $H$ by adding an edge.
Then, there are four possible cases as follows.
\begin{itemize}
\setlength{\itemsep}{0mm}
\item
Suppose that we add an edge between two vertices in the same connected component of $H$ such that every maximum matching covers at least one of the two vertices.
Then, $\varphi(H'') = (i, j, k)$.

\item
Suppose that we add an edge between two vertices in the same connected component $K$ of $H$ such that some maximum matching exposes both of the two vertices.
Then, $\varphi(H'') = (i - 1, j, k)$ or $(i - 1, j - 1, k + 1)$ (when $\mathrm{def}(K) > 2$ or $\mathrm{def}(K) = 2$, respectively).

\item
Suppose that we add an edge between two vertices in different connected components $K_1, K_2$ of $H$ such that every maximum matching covers at least one of the two vertices.
Then, $\varphi(H'') = (i, j - 1, k)$ or $(i, j, k - 1)$ (when $\min(\mathrm{def}(K_1), \mathrm{def}(K_2)) > 0$ or $\min(\mathrm{def}(K_1), \mathrm{def}(K_2)) = 0$, respectively).

\item
Suppose that we add an edge between two vertices in different connected components $K_1, K_2$ of $H$ such that some maximum matching exposes both of the two vertices.
Then, $\varphi(H'') = (i - 1, j - 2, k + 1)$ or $(i - 1, j - 1, k)$ (when $\mathrm{def}(K_1) = \mathrm{def}(K_2) = 1$ or $\max(\mathrm{def}(K_1), \mathrm{def}(K_2)) > 1$, respectively).
\end{itemize}

We then prove the statement for each case.

\medskip\noindent
(i)~ When $\varphi(H) = (0, 0, k)$ $(k \ge 1)$ (including the base case that $H$ is complete), we show $\mathrm{opt}(H) = k - 1$ (Statement~1).
In this case, $H$ has a perfect matching as $\mathrm{def}(H) = 0$, so it is clearly optimal (necessary and sufficient) to make $H$ connected by adding $k - 1$ edges.

\medskip\noindent
(ii)~ When $\varphi(H) = (1, 1, k)$ $(k \ge 0)$, we show $\mathrm{opt}(H) = k + 1$ (Statement~2b).
In this case, $H$ has a unique connected component $K$ with $\mathrm{def}(K) > 0$, which satisfies $\mathrm{def}(K) = 2$.

Let $H'$ be the graph obtained from $H$ by adding an edge between two vertices such that some maximum matching exposes both of them, which are in the same connected component $K$.
Then, $\varphi(H') = (0, 0, k + 1)$, and $\mathrm{opt}(H') = k$ more edges are sufficient by induction hypothesis (Statement~1), and hence $\mathrm{opt}(H) \le \mathrm{opt}(\varphi(H')) + 1 = k + 1$.

Let $H''$ be any graph obtained from $H$ by adding an edge.
Following the observation at the beginning, we have $\varphi(H'') = (1, 1, k)$, $(0, 0, k + 1)$, or $(1, 1, k - 1)$, in which we have $\mathrm{opt}(H'') = k + 1$, $k$, or $k$, respectively, by induction hypothesis (Statement~1 or 2b), and hence $\mathrm{opt}(H) \ge \min_{H''} \mathrm{opt}(\varphi(H'')) + 1 = k + 1$.

\medskip\noindent
(iii)~ When $\varphi(H) = (i, 1, k)$ $(i \ge 2,\ k \ge 0)$, we show $\mathrm{opt}(H) = k + i$ (Statement~2b).
In this case, $H$ has a unique connected component $K$ with $\mathrm{def}(K) > 0$, which satisfies $\mathrm{def}(K) = 2i$.
This case is almost the same as the previous one.
We have $\varphi(H') = (i - 1, 1, k)$ and $\varphi(H'') = (i, 1, k)$, $(i - 1, 1, k)$, or $(i, 1, k - 1)$.
By induction hypothesis (Statement~2b), we obtain $\mathrm{opt}(H) \le k + i$ and $\mathrm{opt}(H) \ge k + i$ as with the previous case.

\medskip\noindent
(iv)~ When $\varphi(H) = (i, 2i, k)$ $(i \ge 1,\ k \ge 0)$, we show $\mathrm{opt}(H) = k + 2i - 1$ (Statement~2a).
In this case, $H$ has $2i$ connected components, each of which has deficiency exactly $1$.

Let $H'$ be the graph obtained from $H$ by adding an edge between two vertices such that some maximum matching exposes both of them, which are in different connected components $K_1, K_2$ with $\mathrm{def}(K_1) = \mathrm{def}(K_2) = 1$.
Then, $\varphi(H') = (i - 1, 2(i - 1), k + 1)$.
By induction hypothesis (Statement~1 or 2a), $\mathrm{opt}(H') = k + 2i - 2$ more edges are sufficient, and hence $\mathrm{opt}(H) \le k + 2i - 1$.

Let $H''$ be any graph obtained from $H$ by adding an edge.
Following the observation at the beginning, we have $\varphi(H'') = (i, 2i, k)$, $(i, 2i - 1, k)$, $(i, 2i, k - 1)$, or $(i - 1, 2(i - 1), k + 1)$, in which we have $\mathrm{opt}(H'') = k + 2i - 1$ or $k + 2i - 2$ by induction hypothesis (Statement~1 or 2a), and hence $\mathrm{opt}(H) \ge k + 2i - 1$.

\medskip\noindent
(v)~ The remaining case is when $\varphi(H) = (i, j, k)$ $(i \ge 2,\ 1 < j < 2i,\ k \ge 0)$.
In this case, $H$ has two different connected components $K_1, K_2$ with $\mathrm{def}(K_1) \ge 1$ and $\mathrm{def}(K_2) \ge 2$.
Note that there are two possible cases: $i < j$ (Statement~2a) and $i \ge j$ (Statement~2b).

Let $H'$ be the graph obtained from $H$ by adding an edge between two vertices such that some maximum matching exposes both of them, one of them is in $K_1$, and the other is in $K_2$.
Then, the resulting component of $H'$ still has a positive deficiency, and hence $\varphi(H') = (i - 1, j - 1, k)$.
No matter in which case we are (i.e., $i < j$ or $i \ge j$), the relation between $i$ and $j$ is preserved and we have $\frac{1}{2}\mathrm{def}(H') = \frac{1}{2}\mathrm{def}(H) - 1$, $c(H') = c(H) - 1$, and $c_0(H') = c_0(H)$.
Thus, by induction hypothesis (Statement 2a or 2b, respectively), the stated number of additional edges in total is indeed sufficient.

Let $H''$ be any graph obtained from $H$ by adding an edge.
Following the observation at the beginning, we have $\varphi(H'') = (i, j, k)$, $(i - 1, j, k)$, $(i - 1, j - 1, k + 1)$, $(i, j - 1, k)$, $(i, j, k - 1)$, $(i - 1, j - 2, k + 1)$, or $(i - 1, j - 1, k)$.
In any case, by induction hypothesis, the stated number of additional edges is necessary as follows, which completes the proof.
\begin{itemize}
\setlength{\itemsep}{0mm}
\item
In any case, we have $\frac{1}{2}\mathrm{def}(H'') + c_0(H'') \ge \frac{1}{2}\mathrm{def}(H) + c_0(H) - 1$ and $c(H'') \ge c(H) - 1$.
Thus, if $H$ and $H''$ are in the same situation (which is applied, Statement 2a or 2b), then the consequence is clear.

\item
Suppose that Statement~2a is applied to $H$ and Statement~2b to $H''$.
In this case, we have $j = i + 1$ and $\varphi(H'') = (i - 1, j - 2, k + 1)$ or $(i, j - 1, k)$.
Then, $\mathrm{opt}(H'') = \frac{1}{2}\mathrm{def}(H'') + c_0(H'') = i + k = j + k - 1 = c(H) - 1$, which means that $c(H) \ (\ge c(H) - 1)$ additional edges in total are necessary in this case.

\item
Suppose that Statement~2b is applied to $H$ and Statement~2a to $H''$.
In this case, we have $j = i$ and $\varphi(H'') = (i - 1, j, k)$.
Then, $\mathrm{opt}(H'') = c(H'') - 1 =  j + k - 1 = i + k - 1 = \frac{1}{2}\mathrm{def}(H) + c_0(H) - 1$, which means that $\frac{1}{2}\mathrm{def}(H) + c_0(H)$ additional edges in total are necessary in this case. \qedhere
\end{itemize}
%
\end{proof}

The proof (the definition of $H'$ in each case) gives a simple greedy algorithm for \textsc{AugPMST} as follows.

\begin{description}
\setlength{\itemsep}{0mm}
\item[Step 0.] Set $\tilde{H} \leftarrow H$, and find a maximum matching $\tilde{M}$ in $\tilde{H}$.
\item[Step 1.] While there exist two different connected components $K_1, K_2$ of $\tilde{H}$ with $\mathrm{def}(K_1) \ge 1$ and $\mathrm{def}(K_2) \ge 2$ (in Case (v)), pick two vertices $v_1, v_2$ exposed by $\tilde{M}$ such that $v_1$ is in $K_1$ and $v_2$ is in $K_2$, and add an edge $\{v_1, v_2\}$ to $\tilde{H}$ and $\tilde{M}$.
\item[Step 2.] While there exist two different connected components $K_1, K_2$ of $\tilde{H}$ with $\mathrm{def}(K_1) = \mathrm{def}(K_2) = 1$ (in Case (iv)), pick two vertices $v_1, v_2$ exposed by $\tilde{M}$ such that $v_1$ is in $K_1$ and $v_2$ is in $K_2$, and add an edge $\{v_1, v_2\}$ to $\tilde{H}$ and $\tilde{M}$.
\item[Step 3.] While $\mathrm{def}(\tilde{H}) > 0$ (in Case (ii) or (iii)), pick two vertices $v_1, v_2$ exposed by $\tilde{M}$ (which are in the same connected component of $\tilde{H}$ by Steps 1 and 2), and add an edge $\{v_1, v_2\}$ to $\tilde{H}$ and $\tilde{M}$.
\item[Step 4.] While $\tilde{H}$ is not connected (in Case (i)), pick two vertices $v_1, v_2$ in different connected components of $\tilde{H}$, and add an edge $\{v_1, v_2\}$ to $\tilde{H}$.
\end{description}

The bottleneck of its computational time is usually finding a maximum matching in $H$, and the other parts are simply implemented in linear time.
Thus, it runs in $\mathrm{O}(\min(\sqrt{n}m + n, n^\omega))$ time \cite{micali1980v,harvey2009algebraic,vazirani2024theory}, where $n$ and $m$ are the numbers of vertices and edges in $H$, respectively, and $\omega < 2.372$ is the matrix multiplication exponent~\cite{duan2023faster}.
For the original problem \textsc{MinPMST} when $G$ is a complete graph and $|\mathrm{cod}(w)| = 2$, if the input is given by specifying which edges have the smaller weight, it runs in $\mathrm{O}(\min(\sqrt{n}m_0 + n, n^\omega))$ time, where $n$ is the number of vertices and $m_0$ is the number of edges having the smaller weight.

\begin{remark}\label{rem:bipartite}
Lemma~\ref{lem:PMST} holds as it is if the underlying graph (i.e., $H$ together with all possible additional edges) is a balanced complete bipartite graph, where the balancedness is necessary to admit a perfect matching.
Also, the above greedy algorithm (with appropriate choices of $v_1$ and $v_2$ as well as $K_1$ and $K_2$ in Steps~1--4) works for \textsc{MinPMST} when $G$ is a balanced complete bipartite graph and $|\mathrm{cod}(w)| = 2$.
\end{remark}

\subsection{NP-Hardness (Statement 2)}\label{sec:PMST_hard}
We refer to a cycle $C$ in a graph as a sequence of vertices $(v_1, v_2, \dots, v_\ell, v_1)$, where the vertices $v_1, v_2, \dots, v_\ell$ $(\ell \ge 3)$ are all distinct and in the graph there exists an edge $\{v_i, v_{i+1}\}$ for each $i = 1, 2, \dots, \ell - 1$ as well as an edge $\{v_\ell, v_1\}$.
A cycle is called \emph{Hamiltonian} if it contains all the vertices in the graph.
The following problem is one of the most fundamental NP-complete problems, which is NP-complete even for cubic planar bipartite graphs.

\begin{problem}[\textsc{Hamiltonian Cycle (HC)}]
\mbox{ }
\begin{description}
\setlength{\itemsep}{0mm}
\item[Input:] A graph $G = (V, E)$.
\item[Goal:] Test whether $G$ has a Hamiltonian cycle or not.
\end{description}
\end{problem}

\begin{theorem}[Akiyama, Nishizeki, and Saito~\cite{akiyama1980np}]
\textsc{HC} is NP-complete even if $G$ is restricted to a cubic planar bipartite graph.
\end{theorem}

We prove Statement~2 of Theorem~\ref{thm:PMST} by reducing this restricted \textsc{HC} to \textsc{MinPMST} with stated restrictions.
Let $G = (V, E)$ be a cubic planar bipartite graph, and fix its planar embedding.
We construct from $G$ a cubic planar bipartite graph $\tilde{G} = (\tilde{V}, \tilde{E})$ with edge weight $\tilde{w} \colon \tilde{E} \to \{0, 1\}$ as follows (see Figure~\ref{fig:PMST}).

For each edge $e = \{u, v\} \in E$, replace $e$ with two disjoint paths of length $3$ between $u$ and $v$; in each path, the only middle edge is of weight $1$ and the other two edges are of weight $0$.
For each vertex $u \in V$, let $f_{u, 1}, f_{u, 2}, f_{u, 3} \in E$ denote the three edges incident to $u$ in this order in the clockwise direction, and then let $u'_1, u'_2, \dots, u'_6$ be the new vertices adjacent to $u$ in this order, where $u'_{2i-1}$ and $u'_{2i}$ are created instead of the edge $f_{u, i}$ for each $i = 1, 2, 3$.
Rename $u$ as $\tilde{u}_0$, merge $u'_4$ and $u'_5$ into a single vertex $\tilde{u}_1$, $u'_6$ and $u'_1$ into $\tilde{u}_2$, and $u'_2$ and $u'_3$ into $\tilde{u}_3$, and remove one of the two parallel edges of weight $0$ between $\tilde{u}_0$ and $\tilde{u}_i$ for each $i = 1, 2, 3$.

\begin{figure}[t!]
\begin{center}
\includegraphics[scale=0.45]{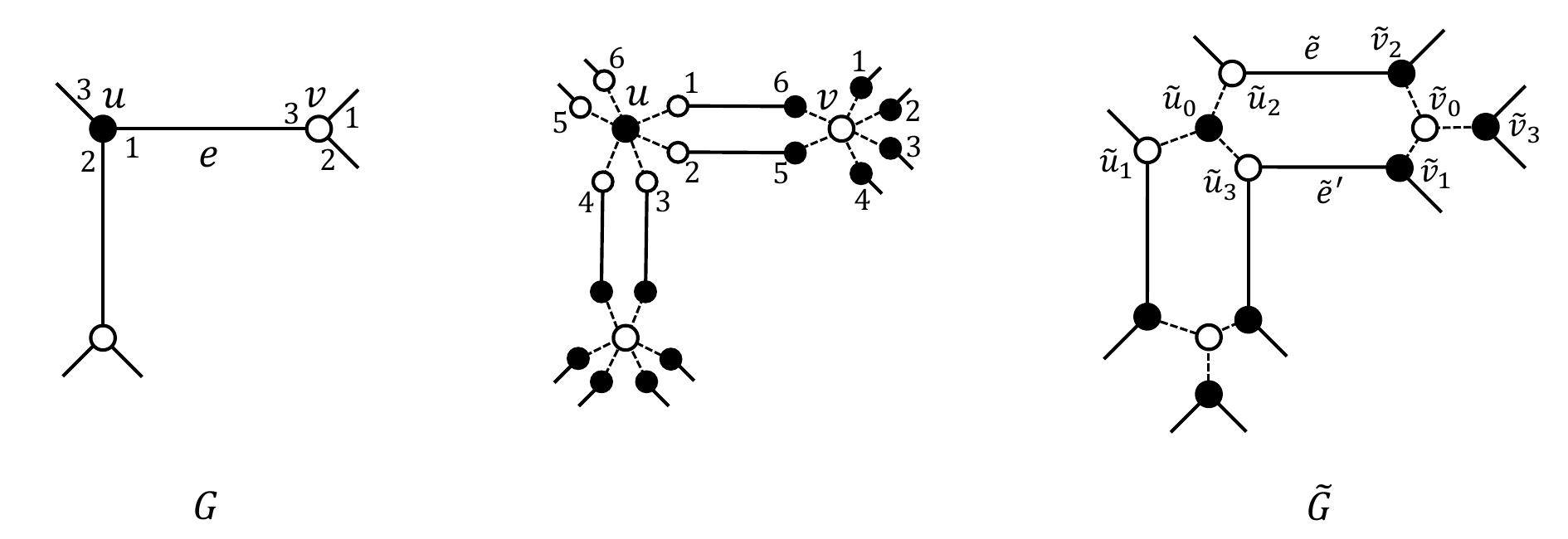}
\caption{Construction of $\tilde{G}$ from $G$. The middle graph is the intermediate graph just after replacing each edge with two disjoint paths. Dashed lines represent edges of weight $0$, and solid lines represent edges of weight $1$. In this example, $\tilde{e}$ and $\tilde{e}'$ are derived from $e = f_{u, 1} = f_{v, 3}$.}
\label{fig:PMST}
\end{center}
\end{figure}

The resulting graph $\tilde{G} = (\tilde{V}, \tilde{E})$ is indeed a cubic planar bipartite graph.
The following claim completes the proof of Statement~2a.

\begin{claim}\label{cl:PMST}
$G$ has a Hamiltonian cycle if and only if $\tilde{G}$ has a spanning tree $\tilde{T}$ containing a perfect matching with $\tilde{w}(\tilde{T}) \le |V|$.
\end{claim}

\begin{proof}
Suppose that $G$ has a Hamiltonian cycle $X$.
By Observation~\ref{obs:PMST}, it suffices to construct a connected subgraph $\tilde{G}' = (\tilde{V}, \tilde{E}')$ such that $\tilde{G}'$ admits a perfect matching and $\tilde{w}(E') \le |V|$.
For this purpose, we can assume that all the edges of weight $0$ are included in $\tilde{E}'$.

Let $\tilde{E}_0$ be the set of edges of weight $0$, and $\tilde{X}$ be the set of edges of weight $1$ each of which is derived from an edge in the Hamiltonian cycle $X$.
We then observe that the subgraph $(\tilde{V}, \tilde{E}_0 \cup \tilde{X})$ is connected since $X$ is a Hamiltonian cycle in $G$.
We also see $\tilde{w}(\tilde{E}_0 \cup \tilde{X}) = |\tilde{X}| = 2|V|$ as exactly two edges in $\tilde{E}$ are derived from each edge in $E$.
Since the two edges derived from the same edge connect the same pair of connected components of $\tilde{E}_0$, even if we remove one of them for each edge in $X$, the resulting subgraph is still connected.
Thus, in order to construct a desired subgraph $\tilde{G}'$, it suffices to choose one of the two edges for each pair so that the chosen edges form a matching in $\tilde{G}$, which can be extended to a perfect matching in $\tilde{G}$ by using edges in $\tilde{E}_0$ (since exactly two edges in $X$ are incident to each $u \in V$, exactly one of $\tilde{u}_1, \tilde{u}_2, \tilde{u}_3$ is exposed by the matching formed by the chosen edges, which can be matched with $\tilde{u}_0$).

Let $X = (x_1, x_2, \dots, x_n, x_1)$ and $s = x_1$.
Without loss of generality, we assume that $f_{s, 1} = \{x_n, x_1\}$ and $f_{s, 2} = \{x_1, x_2\}$ are the two edges around $s$ traversed in this order in $X$ (by shifting the indices of $f_{s, i}$ and by reversing the indices of $X$ if necessary).
Then, the two edges in $\tilde{G}$ corresponding to $f_{s, 1}$ are incident to $\tilde{s}_2$ and $\tilde{s}_3$, and those corresponding to $f_{s, 2}$ are incident to $\tilde{s}_3$ and $\tilde{s}_1$.
First, let us choose the latter edge $\tilde{e}_1$ incident to $\tilde{s}_1$, which is disjoint from both edges corresponding to $f_{s, 1}$.
For the remaining vertices $x_k$ $(k = 2, 3, \dots, n)$, in the ascending order of $k$, we can choose one edge $\tilde{e}_k$ corresponding to $\{x_{k-1}, x_k\}$ so that $\tilde{e}_{k-1}$ and $\tilde{e}_k$ are disjoint (as we always have two disjoint choices of $\tilde{e}_k$).
Recall that both edges corresponding to $f_{s, 1}$ are disjoint from $\tilde{e}_1$, which implies that $\tilde{e}_n$ and $\tilde{e}_1$ are also disjoint (regardless of the choice of $\tilde{e}_n$).
Thus, the chosen edges are pairwise disjoint, i.e., they form a matching in $\tilde{G}$, and we are done.

\medskip
Suppose that $\tilde{G}$ has a spanning tree $\tilde{T}$ containing a perfect matching with $\tilde{w}(\tilde{T}) \le |V|$.
Since $|\tilde{T}| = |\tilde{V}| - 1 = 4|V| - 1$, it consists of at most $|V|$ edges of weight $1$ and at least $3|V| - 1$ edges of weight $0$.
Let $\tilde{T}_1 \coloneq \{\, e \in \tilde{T} \mid w(e) = 1 \,\}$ and $T_1$ be the set of edges in $E$ corresponding to the edges in $\tilde{T}_1$.
Observe that $T_1$ must be connected.

Let $\tilde{M} \subseteq \tilde{T}$ be a perfect matching.
Then, for each vertex $u \in V$, the corresponding center vertex $\tilde{u}_0 \in \tilde{V}$ is matched with an edge $\{\tilde{u}_0, \tilde{u}_i\}$ $(i \in \{1, 2, 3\})$ of weight $0$, and the other two neighbors $\tilde{u}_j$ and $\tilde{u}_k$ $(j \neq i \neq k)$ are matched with edges of weight $1$.
Since there are $2|V|$ such vertices in total and $|\tilde{T}_1| \le |V|$, we must have $\tilde{T}_1 \subseteq \tilde{M}$ and $|\tilde{T}_1| = |V|$.
That is, $T_1$ is a connected spanning subgraph of $G$ such that each vertex has its degree exactly $2$, which is indeed a Hamiltonian cycle.
\end{proof}

For Statement~2b, we add the absent edges of $\tilde{G}$ by setting their weight as $2$.
Then, Claim~\ref{cl:PMST} holds as it is (by replacing $\tilde{G}$ with the augmented graph), because any perfect matching uses at least $|V|$ edges of weight at least $1$.

\section{On Strongly Balanced Spanning Trees (Proof of Theorem~\ref{thm:SBST})}\label{sec:SBST}
In this section, we prove Theorem~\ref{thm:SBST}.
We remark that Lemma~\ref{lem:NY2024} implies the following observation, which leads to the tractability of \textsc{MinSBST} for the bipartite graphs with the aid of polynomial-time algorithms for the weighted matroid intersection problem.

\begin{observation}[cf.~{\cite[Lemma~5]{frank1998bipartite} and \cite[Lemma~3.6]{norose2024approximation}}]\label{obs:SBST}
For a balanced bipartite graph $G$, the set of strongly balanced spanning trees of $G$ can be represented as the set of common bases of two matroids, one of which is graphic and the other is (a truncation of) a partition matroid.
\end{observation}

In what follows, we prove the NP-hardness of \textsc{SBST}.
The \emph{incidence graph} of a $3$-CNF $\psi(x) = C_1 \wedge C_2 \wedge \cdots \wedge C_m$ on $n$ boolean variables $x = (x_1, x_2, \dots, x_n)$ is a bipartite graph defined as follows:
the vertex set is the set of variables and clauses, and an edge exists between a variable $x_i$ and a clause $C_j$ if and only if $C_j$ contains a positive or negative literal of $x_i$.
The 3-SAT problem is known to be NP-complete even when the incidence graph of the input $3$-CNF is very restricted.

\begin{problem}[\textsc{Satisfiability of 3-CNF (3-SAT)}]
\mbox{ }
\begin{description}
\setlength{\itemsep}{0mm}
\item[Input:] A $3$-CNF $\psi(x)$ on $n$ boolean variables $x = (x_1, x_2, \dots, x_n)$.
\item[Goal:] Test whether there exists an assignment $a \in \{0, 1\}^n$ such that $\psi(a) = 1$ or not.
\end{description}
\end{problem}

\begin{theorem}[Lichtenstein~\cite{lichtenstein1982planar}]
\textsc{3-SAT} is NP-complete even if $\psi(x)$ is restricted so that the incidence graph of $\psi(x)$ attached with a Hamiltonian cycle on the variables is planar.
\end{theorem}

We prove Theorem~\ref{thm:SBST} by reducing this restricted \textsc{3-SAT} to \textsc{SBST} with the stated restriction.
Let $\psi(x)$ be a $3$-CNF $\psi(x) = C_1 \wedge C_2 \wedge \cdots \wedge C_m$ on $n$ boolean variables $x = (x_1, x_2, \dots, x_n)$ whose incidence graph attached with a Hamiltonian cycle $X$ on the variables is planar, and fix its planar embedding.
Without loss of generality, $X$ intersects $x_1, x_2, \dots, x_n$ in this order.
We construct from $\psi$ a subcubic planar graph $G = (V, E)$ as follows (see Figures~\ref{fig:SBST_variable}--\ref{fig:SBST_clause}). 

For each variable $x_i$ $(i = 1, 2, \dots, n)$, which appears in $k_i^\In$ clauses lying inside of $X$ and in $k_i^\Out$ clauses lying outside of $X$ (under the planar embedding of the incidence graph fixed above), create the following variable gadget (see Figure~\ref{fig:SBST_variable}).
Create a cycle
\[(u_i, u^\In_{i, 0}, u^\In_{i, 1}, \bar{u}^\In_{i, 1}, \dots, u^\In_{i, k_i^\In}, \bar{u}^\In_{i, k_i^\In}, u_i^\mathrm{T}, u_i^\mathrm{F}, u^\Out_{i, k_i^\Out}, \bar{u}^\Out_{i, k_i^\Out}, \dots, u^\Out_{i, 1}, \bar{u}^\Out_{i, 1}, u^\Out_{i, 0}, u_i).\]
Then, add two vertices $u_{i, 0}$ and $u_i^\End$ with four incident edges $\{u_{i, 0}, u^\In_{i, 0}\}$, $\{u_{i, 0}, u^\Out_{i, 0}\}$, $\{u_i^\End, u_i^\mathrm{T}\}$, and $\{u_i^\End, u_i^\mathrm{F}\}$, and two vertices $u'_{i,0}$ and $u''_{i,0}$ with two incident edges $\{u_{i, 0}, u'_{i, 0}\}$ and $\{u'_{i, 0}, u''_{i, 0}\}$.

\begin{figure}[t!]
\begin{center}
\includegraphics[scale=0.4]{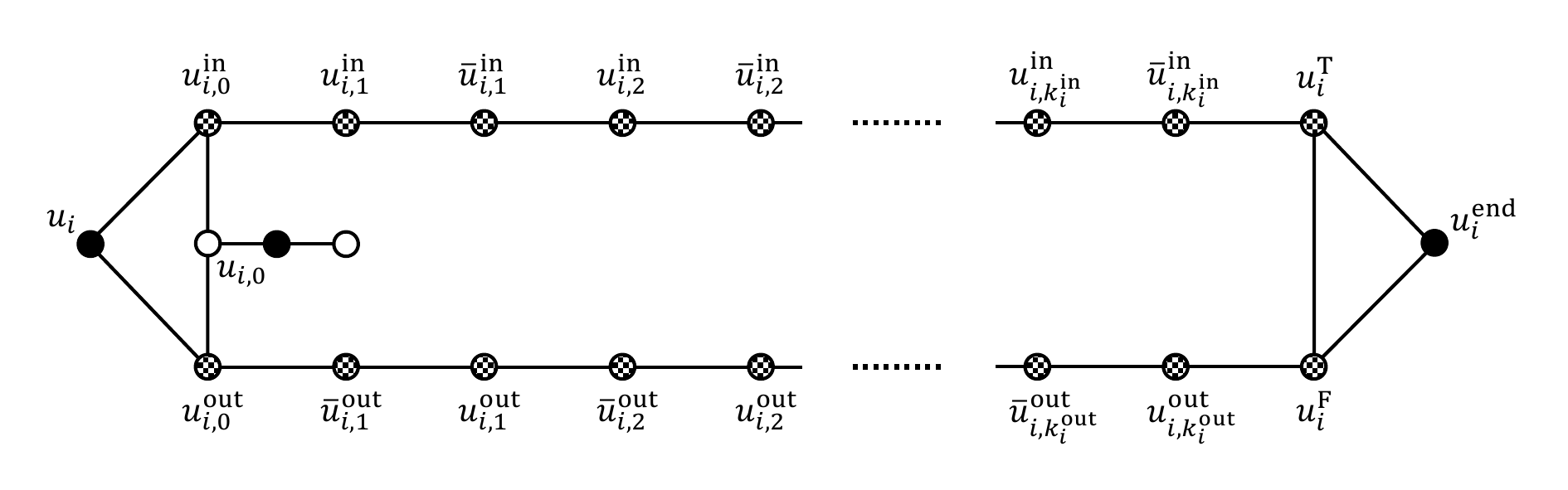} \\
\includegraphics[scale=0.4]{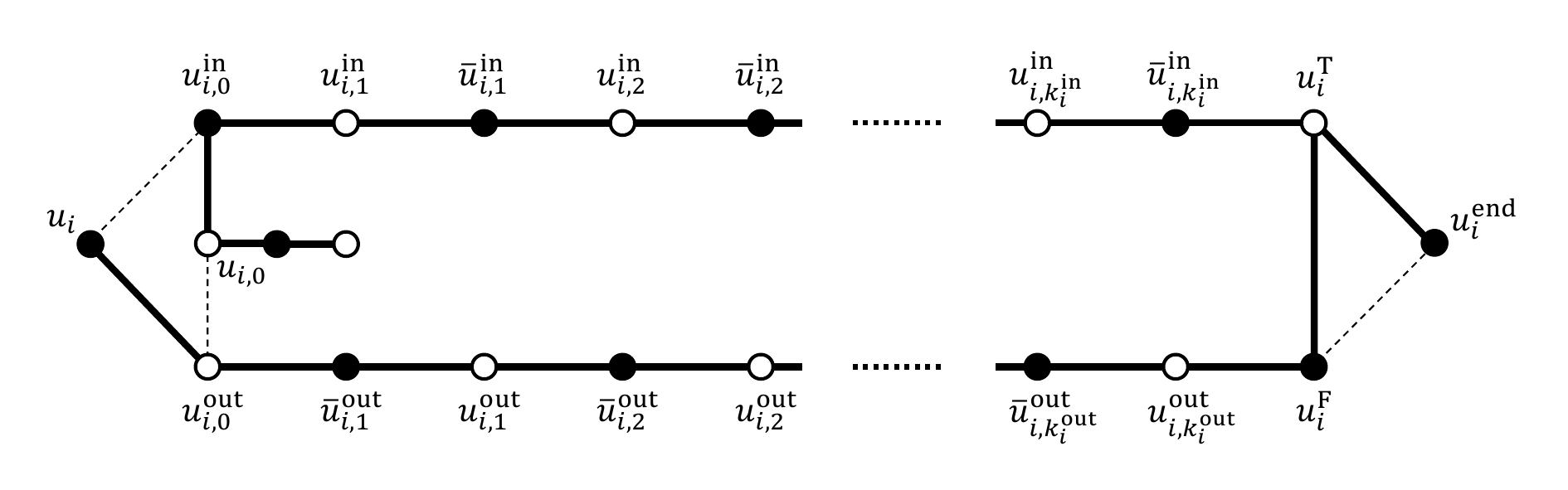}\vspace{-1mm}
\caption{The variable gadget for a variable $x_i$. In any strongly balanced spanning tree $T$ (if exists), the black vertices will be in $V_T^+$ (with degree constraint $2$), the white vertices will be in $V_T^-$ (without degree constraint), and the checkered vertices can be in either side, alternately from $u_i$; the lower figure illustrates a subtree corresponding to an assignment with $x_i = 1$.}
\label{fig:SBST_variable}
\end{center}
\end{figure}

We connect those variable gadgets as Figure~\ref{fig:SBST_path}.
Specifically, for each $i = 1, 2, \dots, n - 1$, we introduce a joint vertex $u_{i, i+1}$ with two incident edges $\{u_i^\mathrm{end}, u_{i, i+1}\}$ and $\{u_{i, i+1}, u_{i+1}\}$, and put a vertex $u_{n, n+1}$ with an incident edge $\{u_n^\mathrm{end}, u_{n, n+1}\}$ at the end.
Furthermore, at the beginning (before $u_1$), we put a tree consisting of eight vertices as illustrated, whose three leaves are named as $s_1, s_2, s_3$.

\begin{figure}[t!]
\begin{center}
\includegraphics[scale=0.4]{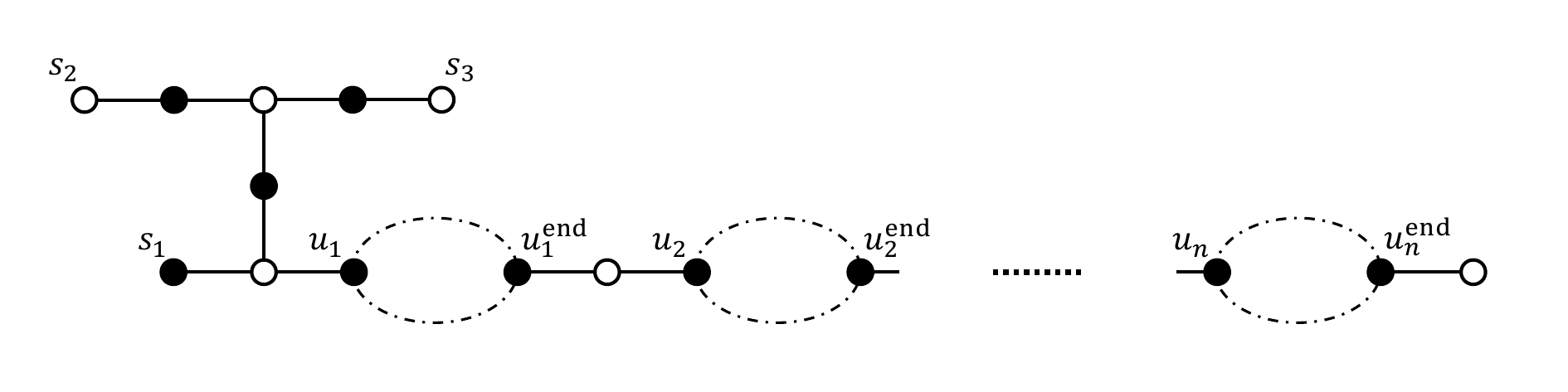}
\caption{The whole structure of the variable gadgets. In any strongly balanced spanning tree $T$ (if exists), the black vertices will be in $V_T^+$ (with degree constraint $2$ except for $s_1$, which will be the unique leaf in $V_T^+$) and the white vertices will be in $V_T^-$ (without degree constraint).}
\label{fig:SBST_path}
\end{center}
\end{figure}

Finally, for each clause $C_j = (y_{j,1} \vee y_{j,2} \vee y_{j,3})$, create a clause gadget as follows (see Figure~\ref{fig:SBST_clause}).
Create a cycle $(c_{j,1}, c_{j,12}, c_{j,2}, c_{j,23}, c_{j,3}, c_{j,31}, c_{j,1})$, and add two vertices $c_j$ and $c'_j$ with two incident edges $\{c_j, c'_j\}$ and $\{c'_j, c_{j,31}\}$.
For each $\ell = 1, 2, 3$, add an edge between $c_{j,\ell}$ and $b_{j,\ell}$, where $b_{j,\ell}$ is a vertex in a variable gadget that is determined depending on whether $C_j$ lies inside or outside of $X$ and what literal $y_{j,\ell}$ is as follows.
Suppose that $C_j$ lies inside of $X$, and that $y_{j,\ell}$ is a positive literal of a variable $x_i$.
Then, $b_{j,\ell} = u_{i,k}^\In$, where $1 \le k \le k_i^\In$ is such that $y_{j,\ell}$ is the $k$-th appearance of $x_i$ in clauses lying inside of $X$ along the cycle $X = (x_1, x_2, \dots, x_n, x_1)$.
The other three cases are analogous; if $C_j$ lies outside of $X$, then replace the superscripts ${\bullet}^\In$ with ${\bullet}^\Out$, and if $y_{j,\ell}$ is a negative literal of $x_i$, then replace $u_i$ with $\bar{u}_i$.

\begin{figure}[t!]
\begin{center}\hspace{-10mm}
\begin{minipage}[b]{0.48\textwidth}
\centering
\includegraphics[scale=0.36]{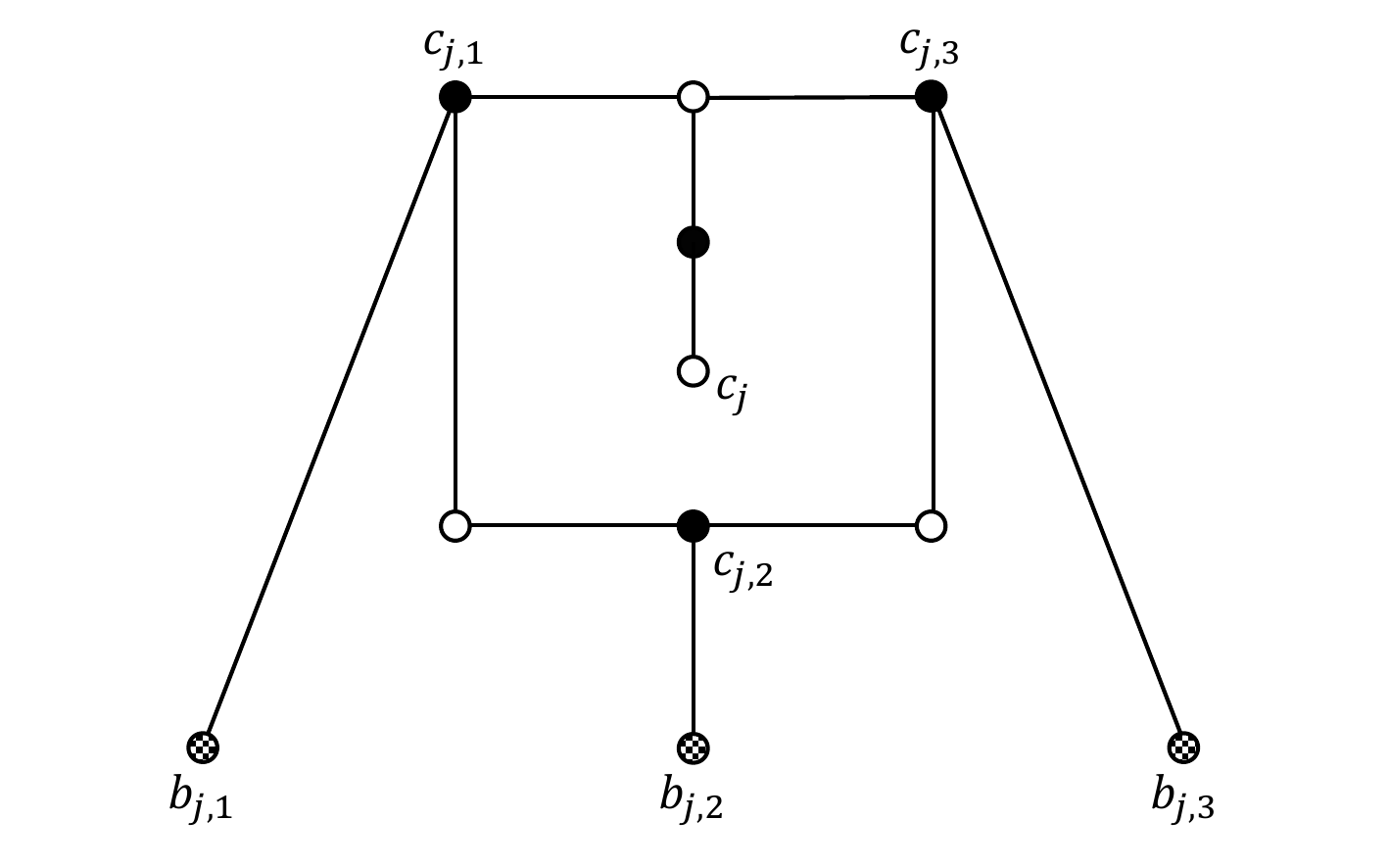}
\end{minipage}\hspace{7mm}
\begin{minipage}[b]{0.48\textwidth}
\centering
\includegraphics[scale=0.36]{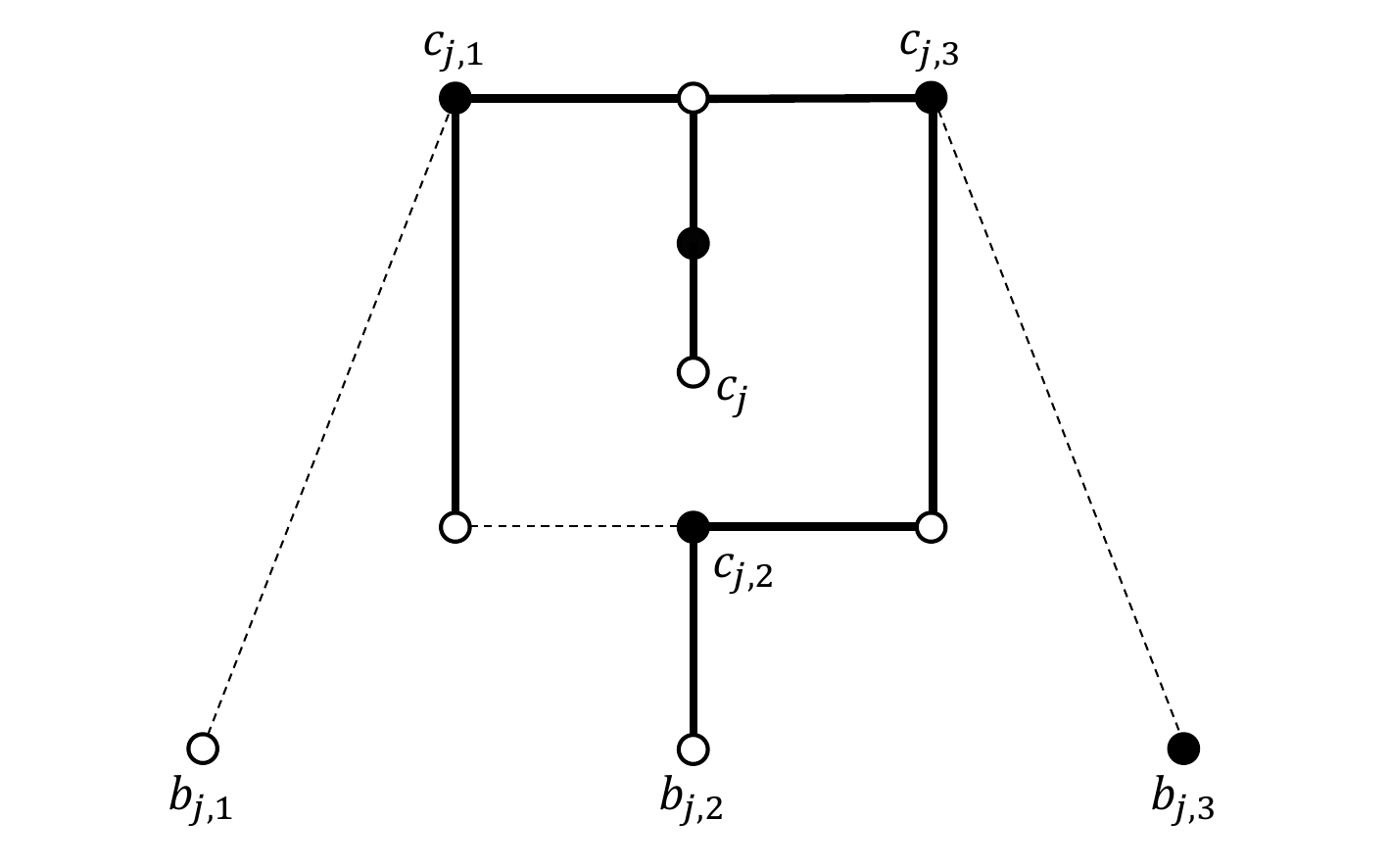}
\end{minipage}
\caption{The clause gadget for a clause $C_j = (y_{j,1} \vee y_{j,2} \vee y_{j,3})$. In any strongly balanced spanning tree $T$ (if exists), the black vertices will be in $V_T^+$ (with degree constraint $2$), the white vertices will be in $V_T^-$ (without degree constraint), and the checkered vertices can be in either side depending on the situation in the corresponding variable gadgets; the right figure illustrates an example of a subtree corresponding to the situation when $C_j$ is satisfied by $y_{j,2} = 1$.}
\label{fig:SBST_clause}
\end{center}
\end{figure}

The resulting graph $G = (V, E)$ is clearly subcubic and planar.
The following claim completes the proof of Theorem~\ref{thm:SBST}.

\begin{claim}\label{cl:SBST}
$\psi$ is satisfiable if and only if $G$ has a strongly balanced spanning tree.
\end{claim}

\begin{proof}
Suppose that $a = (a_1, a_2, \dots, a_n) \in \{0, 1\}^n$ satisfies $\psi$.
Then, we can construct a strongly balanced spanning tree $T$ of $G$ as follows (see also Figures~\ref{fig:SBST_variable}--\ref{fig:SBST_clause}).

For each variable $x_i \ (= a_i)$, we construct a spanning tree in the corresponding variable gadget by deleting three edges;
we delete $\{u_i, u_{i,0}^\In\}$, $\{u_{i, 0}, u_{i,0}^\Out\}$, and $\{u_i^\mathrm{F}, u_i^\End\}$ if $a_i = 1$ (cf.~Figure~\ref{fig:SBST_variable}), and we delete $\{u_i, u_{i,0}^\Out\}$, $\{u_{i, 0}, u_{i,0}^\In\}$, and $\{u_i^\mathrm{T}, u_i^\End\}$ if $a_i = 0$.
Under the assumption that $u_i \in V_T^+$, in the former case (when $a_i = 1$), all the vertices $u_{i,k}^\bullet$ $(k \ge 1,\ \bullet \in \{\In, \Out\})$ and $u_i^\mathrm{T}$ are in $V_T^-$ (i.e., without degree constraint), and all the vertices in the form $\bar{u}_{i,k}^\bullet$ $(k \ge 1,\ \bullet \in \{\In, \Out\})$ and $u_i^\mathrm{F}$ are in $V_T^+$ with exactly two incident edges in this gadget; in the latter case (when $a_i = 0$), $V_T^+$ and $V_T^-$ are interchanged.
Also, in either case, we have $u_i^\End \in V_T^+$ (under the assumption that $u_i \in V_T^+$).

We connect those spanning trees of the variable gadgets by taking all the edges connecting them including the tree at the beginning and the leaf at the end (illustrated in Figure~\ref{fig:SBST_path}).
Here, note that $s_1, s_2, s_3$ must be leaves in $T$, and $s_2$ and $s_3$ are on the same side of the bipartition $(V_T^+, V_T^-)$, which is different from $s_1$.
This forces $s_1$ to be a unique leaf in $V_T^+$.
As a result, inductively, $u_i$ and $u_i^\End$ $(i = 1, 2, \dots, n)$ are all in $V_T^+$.

Finally, for each clause $C_j = (y_{j, 1} \vee y_{j, 2} \vee y_{j, 3})$, pick one of the literals $y_{j, \ell}$ whose value is $1$ in the assignment $x = a$.
Then, take an edge $\{b_{j, \ell}, c_{j, \ell}\}$ and all but one edge incident to $c_{j, \ell}$ in the corresponding clause gadget (cf.~Figure~\ref{fig:SBST_clause}).
This results in a spanning tree of the clause gadget connected to the variable gadgets with the bridge $\{b_{j, \ell}, c_{j, \ell}\}$.
By the above observation and assumption, we have $b_{j, \ell} \in V_T^-$.
We then have $c_{j,1}, c_{j,2}, c_{j,3}, c'_j \in V_T^+$, and these vertices indeed have degree $2$ in the resulting tree.

Overall, we have indeed constructed a desired, strongly balanced spanning tree of $G$.

\medskip
Suppose that $G$ has a strongly balanced spanning tree $T$.
We show that $T$ should be in the form constructed above, which implies that we can reconstruct an assignment $a = (a_1, a_2, \dots, a_n)$ satisfying $\psi$.

The main task is to confirm that for any clause gadget plus its three neighbors $b_{j,1}, b_{j,2}, b_{j,3}$, the restriction of $T$ there does not contain a path between two of the neighbors; that is, any clause gadget does not play the role of connecting variable gadgets.
Observe that $T$ must contain the two edges incident to $c'_j$, and $c_j \in V_T^-$, $c'_j \in V_T^+$, and $c_{j, 31} \in V_T^-$.
By connectivity, at least one of $\{c_{j,1}, c_{j,31}\}$ and $\{c_{j,3}, c_{j,31}\}$ is in $T$.

Suppose that there exists a path between $b_{j,1}$ and $b_{j,3}$.
Since $\{c_{j,1}, c_{j,31}\}$ or $\{c_{j,3}, c_{j,31}\}$ is in $T$ and then $c_{j,1}$ or $c_{j,3}$, respectively, is in $V_T^+$, the path must be $(b_{j,1}, c_{j,1}, c_{j,31}, c_{j,3}, b_{j,3})$ and $c_{j,1}, c_{j,3} \in V_T^+$.
Then, neither $\{c_{j,1}, c_{j,12}\}$ nor $\{c_{j,3}, c_{j,23}\}$ is in $T$, and hence $\{c_{j,2}, c_{j,12}\}$ and $\{c_{j,2}, c_{j,23}\}$ must be in $T$.
Also, as $T$ is connected, $\{c_{j,2}, b_{j,2}\}$ must be in $T$.
This, however, cannot satisfy the degree constraint no matter which $c_{j,12}$ (degree $1$) or $c_{j,2}$ (degree $3$) is in $V_T^+$, a contradiction.

Suppose that there exists a path between $b_{j,1}$ and $b_{j,2}$.
If the path is via $c_{j,13}$, then $c_{j,1}, c_{j,2} \in V_T^+$ and hence neither of $\{c_{j,1}, c_{j,12}\}$ and $\{c_{j,2}, c_{j,12}\}$ is in $T$; then $c_{j,12}$ is isolated, a contradiction.
Otherwise, the path is $(b_{j,1}, c_{j,1}, c_{j,12}, c_{j,2}, b_{j,2})$.
In this case, $T$ must contain $\{c_{j,3}, c_{j,31}\}$ and hence $c_{j,3} \in V_T^+$.
If $\{c_{j,3}, c_{j,23}\}$ is in $T$, then $\{c_{j,3}, b_{j,3}\}$ cannot be in $T$ by the degree constraint, and hence $\{c_{j,2}, c_{j,23}\}$ is also in $T$ as $T$ is connected.
This, however, violates the degree constraint of $c_{j,2} \in V_T^+$, a contradiction.
Otherwise, $\{c_{j,3}, c_{j,23}\}$ is not in $T$, and then $\{c_{j,2}, c_{j,23}\}$ is in $T$ again as $T$ is connected.
This, however, cannot satisfy the degree constraint no matter which $c_{j,2}$ (degree $3$) or $c_{j,23}$ (degree $1$) is in $V_T^+$, a contradiction.

Thus, we have confirmed that any clause gadget does not connect variable gadgets.
Note that $T$ may contain two or three of $\{c_{j,1}, b_{j,1}\}$, $\{c_{j,2}, b_{j,2}\}$, and $\{c_{j,3}, b_{j,3}\}$, and then the end vertex $b_{j,\ell}$ with $\{c_{j,\ell}, b_{j,\ell}\}$ contained in $T$ must be in $V_T^-$ due to the degree constraint.
Also, no matter how many such edges are contained in $T$, exactly one of them is extended to $c_j$.

Next, let us consider variable gadgets.
By the above observation, they must be connected with the edges illustrated in Figure~\ref{fig:SBST_path}, which inductively forces that $u_i$ and $u_i^\End$ $(i = 1, 2, \dots, n)$ are both in $V_T^+$ as follows.

For each variable gadget, due to the form of the eight-vertex tree at the beginning $(i = 1)$ or by the induction hypothesis $(i \ge 2)$, we have $u_i \in V_T^+$ and $T$ contains exactly one edge incident to $u_i$ not in the variable gadget.
Then, exactly one of $\{u_i, u_{i,0}^\In\}$ and $\{u_i, u_{i,0}^\Out\}$ must be contained in $T$, and then $u_{i,0}^\In$ or $u_{i,0}^\Out$, respectively, is in $V_T^-$.
As with the clause gadgets, observe that $T$ must contain the two edges incident to $u'_{i,0}$, and $u''_{i, 0} \in V_T^-$, $u'_{i, 0} \in V_T^+$, and $u_{i, 0} \in V_T^-$.
By connectivity, $\{u_{i,0}, u_{i,0}^\In\}$ or $\{u_{i,0}, u_{i,0}^\Out\}$ is in $T$, and then $u_{i,0}^\In$ or $u_{i,0}^\Out$, respectively, is in $V_T^+$.
Thus, we have exactly two possible choices here such that exactly one of $u_{i,0}^\In$ and $u_{i,0}^\Out$ is in $V_T^+$ and the other is in $V_T^-$.
By connectivity again, the two paths $(u_{i,0}^\In, u_{i,1}^\In, \bar{u}_{i,1}^\In, \dots, u_{i,k_i^\In}^\In, \bar{u}_{i,k_i^\In}^\In, u_i^\mathrm{T})$ and $(u_{i,0}^\Out, \bar{u}_{i,1}^\Out, u_{i,1}^\Out, \dots, \bar{u}_{i,k_i^\Out}^\Out, u_{i,k_i^\Out}^\Out, u_i^\mathrm{F})$ are completely included in $T$.
Due to the parity, $T$ cannot contain both edges $\{u_i^\mathrm{T}, u_i^\End\}$ and $\{u_i^\mathrm{F}, u_i^\End\}$, and hence exactly one of them in addition to $\{u_i^\mathrm{T}, u_i^\mathrm{F}\}$ is contained in $T$.
Then, due to the degree constraint, we obtain $u_i^\End \in V_T^+$, which forces $u_{i, i+1} \in V_T^-$ and $u_{i+1} \in V_T^+$ (when $i < n$).

Overall, in the variable gadget, there are exactly two possible spanning trees, from which an assignment $a$ with $\psi(a) = 1$ can be reconstructed.
This completes the proof.
\end{proof}

\section{Concluding Remarks}\label{sec:conclusion}
We have investigated two problems on a fusion of two fundamental combinatorial structures, a spanning tree and a perfect matching.

The first problem, finding a minimum-weight spanning tree containing a perfect matching, has been shown as tractable in the very restricted situation when the graph is complete (or complete bipartite) and the edge weights take at most two values.
It, however, becomes NP-hard if we relax one of the two conditions.
For this problem, it seems reasonable to consider nontrivial approximation or fixed-parameter algorithms.

The second problem, testing the existence of a strongly balanced spanning tree, has turned out NP-hard even if the input graph is subcubic and planar.
In the reduction, we have introduced many artificial leaves, which have played the important role to force which vertices should be on which side of the resulting bipartition.
This can be somewhat relaxed by replacing each leaf with a five-vertex gadget as in Figure~\ref{fig:SBST_nonleaf} so that the resulting graph is still subcubic and planar and has no leaf.
An interesting open question is the following: is it possible to strengthen ``subcubic'' to ``cubic''? Or, possibly, is this problem for the cubic graphs in fact tractable?
It seems also interesting to consider the tractability for relatively dense graphs, which tend to have a solution.

\begin{figure}[ht!]
\begin{center}
\includegraphics[scale=0.4]{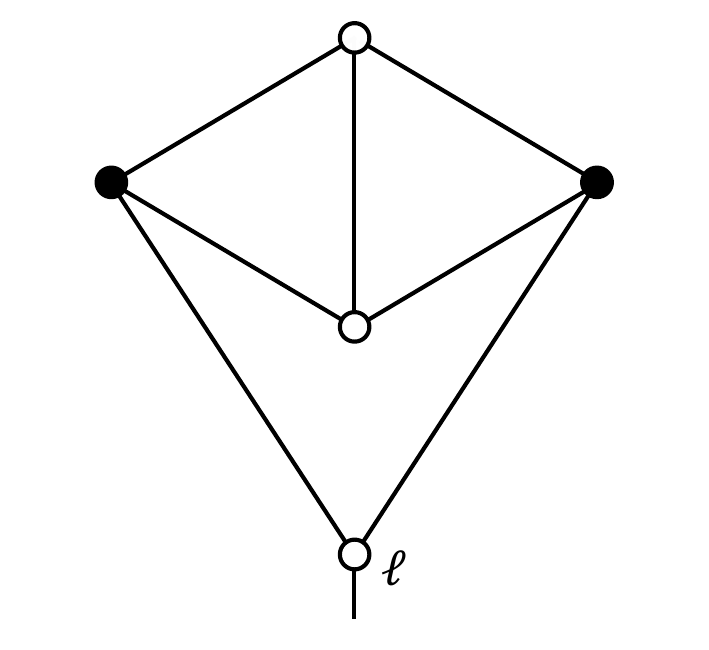}
\caption{An alternative five-vertex gadget for each leaf. The lower vertex $\ell$ corresponds to an original leaf. In any strongly balanced spanning tree $T$ (if exists), the black vertices will be in $V_T^+$ (with degree constraint $2$), and the white vertices will be in $V_T^-$ (without degree constraint), except for exactly one of these gadgets (that corresponds to $s_1$ in Figure~\ref{fig:SBST_path}).}
\label{fig:SBST_nonleaf}
\end{center}
\end{figure}

\bibliographystyle{plain}
\bibliography{main}

\end{document}